\documentclass[runningheads]{llncs}
\usepackage[T1]{fontenc}

\usepackage[hidelinks]{hyperref}
\hypersetup{bookmarksnumbered=true,bookmarksopenlevel=2}

\usepackage{graphicx}
\usepackage{color}

\usepackage{enumitem}
\spnewtheorem{observation}{Observation}{\itshape}{\rmfamily}
\spnewtheorem*{proofsketch}{Proof Sketch}{\itshape}{\rmfamily}

\usepackage{thm-restate}
\usepackage{amsmath,amssymb}
\usepackage[capitalize]{cleveref}
\crefname{observation}{Observation}{Observations}
\crefname{problem}{Problem}{Problems}

\newcommand{\N}{\mathbb{N}}
\newcommand{\exponent}{\mathsf{exponent}}
\newcommand{\mantissa}{\mathsf{mantissa}}
\newcommand{\abs}[1]{\left\vert#1\right\vert}

\newcommand{\target}{Y}

\newcommand{\period}{\ensuremath{P}}
\newcommand{\round}[1]{\left[#1\right]}
\newcommand{\ceil}[1]{\left\lceil#1\right\rceil}
\newcommand{\st}{\text{ s.t. }}
\newcommand{\Sme}{S_{me}}
\newcommand{\Sfeeder}{S_{F}}

\usepackage{scalerel}
\usepackage{tikz}
\usetikzlibrary{svg.path}

\definecolor{orcidlogocol}{HTML}{A6CE39}
\tikzset{
  orcidlogo/.pic={
    \fill[orcidlogocol] svg{M256,128c0,70.7-57.3,128-128,128C57.3,256,0,198.7,0,128C0,57.3,57.3,0,128,0C198.7,0,256,57.3,256,128z};
    \fill[white] svg{M86.3,186.2H70.9V79.1h15.4v48.4V186.2z}
                 svg{M108.9,79.1h41.6c39.6,0,57,28.3,57,53.6c0,27.5-21.5,53.6-56.8,53.6h-41.8V79.1z M124.3,172.4h24.5c34.9,0,42.9-26.5,42.9-39.7c0-21.5-13.7-39.7-43.7-39.7h-23.7V172.4z}
                 svg{M88.7,56.8c0,5.5-4.5,10.1-10.1,10.1c-5.6,0-10.1-4.6-10.1-10.1c0-5.6,4.5-10.1,10.1-10.1C84.2,46.7,88.7,51.3,88.7,56.8z};
  }
}

\renewcommand\orcidID[1]{\href{https://orcid.org/#1}{\mbox{\scalerel*{
\begin{tikzpicture}[yscale=-1,transform shape]
\pic{orcidlogo};
\end{tikzpicture}
}{|}}}}

\begin{document}
\title{\texorpdfstring{Model Checking Linear Dynamical Systems\\ under Floating-point Rounding}{Model Checking Linear Dynamical Systems under Floating-point Rounding}
}

\author{Engel Lefaucheux\inst{1}\orcidID{0000-0003-0875-300X} \and
Jo\"el Ouaknine\inst{2}\orcidID{0000-0003-0031-9356} \and
David Purser\inst{3,4}\orcidID{0000-0003-0394-1634} \and
Mohammadamin Sharifi\inst{5}\orcidID{0000-0002-1987-9487}}
\authorrunning{E. Lefaucheux et al.}

\institute{
University of Lorraine, CNRS, Inria, LORIA, Nancy, France  \email{engel.lefaucheux@inria.fr}
\and
Max Planck Institute for Software Systems, Saarland Informatics Campus, Saarbrücken, Germany \email{joel@mpi-sws.org} \and
University of Warsaw, Warsaw, Poland \and
University of Liverpool, Liverpool, UK \email{D.Purser@liverpool.ac.uk} \and
Sharif University of Technology, Tehran, Iran \email{sharifim689@gmail.com}
}

\maketitle              %

\begin{abstract}
We consider linear dynamical systems under floating-point rounding. 
In these systems, a matrix is repeatedly applied to a vector, but the numbers are rounded into floating-point representation after each step (i.e., stored as a fixed-precision mantissa and an exponent). The approach more faithfully models realistic implementations of linear loops, compared to the exact arbitrary-precision setting often employed in the study of linear dynamical systems.

Our results are twofold: We show that for non-negative matrices there is a special structure to the sequence of vectors generated by the system: the mantissas are periodic and the exponents grow linearly. We leverage this to show decidability of $\omega$-regular temporal model checking against semialgebraic predicates. This contrasts with the unrounded setting, where even the non-negative case encompasses the long-standing open Skolem and Positivity problems.

On the other hand, when negative numbers are allowed in the matrix, we show that the reachability problem is undecidable by encoding a two-counter machine. Again, this is in contrast with the unrounded setting where point-to-point reachability is known to be decidable in polynomial time. 

\end{abstract} 

\section{Introduction}

Loops are a fundamental staple of any programming language, and the study of loops plays a pivotal role in many subfields of computer science, including automated verification, abstract interpretation, program analysis, semantics, etc. 
The focus of the present paper is on the algorithmic analysis of
simple (i.e., non-nested) linear (or affine) while loops, such as the following:
\newpage\begin{verbatim}
x = 3, y = 4, z = 2
while x+3y+z > 4:
    x = 3x +2z
    y = 3x + y
    z = y + z
\end{verbatim}

We are interested in analysing how the loop evolves. A simple reachability query is to decide whether the loop variables ever satisfy a Boolean combination of polynomial inequalities, for example modelling a loop guard.
More generally, one might seek to consider significantly more complex temporal properties, such as those expressible in linear temporal logic or monadic second-order logic: this gives rise to a model-checking problem.

Modelling the evolution of such a loop may require unbounded memory.
That is, the number of bits needed to represent the numbers $x$, $y$, and $z$ may grow 
larger and larger. 
However, most computer systems do not represent rational numbers to arbitrary precision, but rather use \emph{floating-point rounding}, in which a number $y$ is stored using two components: the mantissa $m\in\mathbb{Q}$ 
and the exponent $\alpha\in\mathbb{Z}$, such that $y= m \cdot 10^\alpha$.\footnote{%
We work in base 10 throughout for simplicity of exposition. 
All our results carry over \emph{mutatis mutandis} in any integer base, including base 2 as typically used in practice.} 

Typically floating-point numbers are specified using either 32 or 64 bits, with some of these reserved for the mantissa and some for the exponent, thus bounding both the mantissa and the exponent. 
\textbf{We do not do this}, and only place a bound on the number of bits representing the 
mantissa, allowing the exponent to grow unboundedly (in either direction). From a theoretical standpoint,
bounding the number of bits of both the 
mantissa and the exponent would necessarily give rise to a finite-state system, for which essentially any decision problem would become decidable (at least in principle, if not necessarily in practice). 
Due to the unboundedness of exponents in our setting, we do not have to consider overflows (`NaN', `infinity' or `-infinity' which are part of most floating-point specifications).

Formally, we model our programs using linear dynamical systems (LDS), which comprise
a starting vector representing the initial state of each variable
and a matrix describing the evolution of the program.
An LDS generates an infinite sequence of vectors (the \emph{orbit} of the system) by
multiplying the matrix with the current 
vector and then applying floating-point rounding to the result.

\subsection*{Our results}

We consider the \emph{model-checking} problem for linear dynamical systems evolving under floating-point rounding. More formally, let $\target_1,\dots, \target_k \subseteq \mathbb{R}^d$ be semialgebraic targets.
Given an orbit $(x^{(t)})_{t\in \N}$, we define the characteristic word 
$w = w_1,w_2,w_3,\dots$ with respect to $\target_1,\dots,\target_k$ over alphabet $2^{\{1,\dots, k\}}$ such that $i \in w_t$ if and only if $x^{(t)} \in Y_i$. The model-checking problem asks whether $w$ is in an $\omega$-regular language, or equivalently satisfies a temporal specification given in monadic second-order logic (MSO).

Our results show that analysing LDS under floating-point rounding is neither clearly easier nor harder than in the standard setting (without rounding).  Our first contribution establishes \emph{undecidability} of point-to-point reachability (and \emph{a fortiori} model checking) under floating-point rounding, a surprising outcome given that point-to-point reachability is solvable
in polynomial time without rounding~\cite{KannanL86}. On the other hand,
in the standard setting neither decidability nor undecidability are known for full model checking (although mathematical hardness results exist); see~\cite{DBLP:conf/soda/OuaknineW14,KarimovLOPVWW22,DBLP:conf/birthday/KarimovKO022}. 

\begin{restatable}{theorem}{thmundec}\label{thm:undec}
The floating-point point-to-point reachability problem is undecidable.
\end{restatable}

However, for non-negative matrices, we show that the full MSO model-checking problem is decidable in our setting, without restrictions on the dimensions of the predicates or the ambient space. 
This is in stark contrast to the standard setting, where assuming non-negativity does
not simplify the problem.
Model checking non-negative LDS without rounding would
require (at a minimum) solving the longstanding open Skolem and Positivity problems~\cite{AkshayAOW15}.

\begin{restatable}{theorem}{thmmodelchecking}\label{thm:modelcheckingdec}
Let $(M,x)$ be a non-negative linear dynamical system, let $\target_1,\dots,\target_k$ be semialgebraic targets and let $\phi$ be an MSO formula using predicates over $\target_1,\dots, \target_k   $. It is decidable whether the characteristic word under floating-point rounding satisfies $\phi$.
\end{restatable}

We place no dimension restriction on the predicates; in particular, showing that the Skolem and Positivity problems are \emph{decidable} on non-negative systems under floating-point rounding. At this time we do not however have complexity upper bounds on our model-checking algorithm, or lower bounds on the model-checking problem.

\subsection*{Related work}
There is a line of practical tools for the analysis, verification, and invariant synthesis for floating-point loops~\cite{BeckerPDT18,LoharJSDC21,AbbasiSDUA21,maurica2017optimal}. These tools typically work well in practice, but do not necessarily work in all cases. The analysis of concrete implementations of floating-point specifications requires careful analysis of edge cases around $\pm \infty$ and `NaN'.
In contrast to these tools which focus primarily on practical analysis, our work seeks to understand the theoretical possibilities and limitations of the exact analysis of (possibly long-running) floating-point loops in a generalised setting.

The study of linear dynamical systems explores the sequence of vectors induced by a matrix. 
Model checking is only known to be decidable for certain classes of semialgebraic predicates---in particular those with low dimension~\cite{KarimovLOPVWW22} or for prefix-independent properties~\cite{AlmagorKKO021}; see also~\cite{DBLP:conf/birthday/KarimovKO022}. 
The well-known Skolem and Positivity problems being special cases of model checking,
they place technical limits on the dimensions that can be handled without first resolving 
long-standing open cases of these problems. Recent progress suggests that the Skolem problem may be yet be conquered, at least for diagonalisable matrices~\cite{BiluLNOPW22,LucaOW22}, but Positivity requires solving 
particularly difficult problems in analytic number theory~\cite{DBLP:conf/soda/OuaknineW14,ChonevOW16}.
The non-negative case can be used to model sequences of distributions induced by Markov chains~\cite{baierconcur22}, although all hardness limitations apply already in the probabilistic setting~\cite{AkshayAOW15}.

Baier et al.~\cite{Baier0JKLOPPW20} consider LDS under rounding to fixed-decimal precision, showing reachability is PSPACE-complete for hyperbolic systems (when no eigenvalue has modulus one) and decidable for certain other constrained classes of rounding. A notable difference of fixed-decimal precision is that it cannot allow arbitrarily small numbers, unlike the floating-point numbers we consider.

A recent line of work focusses on linear dynamical systems with perturbations at every step, with a view to understanding the robustness of reachability problems~\cite{DCostaKMOSS021,DCostaKMOSW22,0001BGV22}. However, unlike rounding, the perturbation is chosen in order to assist hitting the target and the perturbation is arbitrarily small. 

For linear while loops the reachability problem can be rephrased as a halting problem, asking whether a guard condition is eventually met from a given initial state. The related termination problem asks whether a guard condition is met from \emph{every} initial state~\cite{Tiwari04,Braverman06}. Issues arising from implementations using floating-point representations to solve the termination problem of unrounded (arbitrary precision) loops are considered in~\cite{XiaYZZ11}. In contrast, we are interested in analysing programs in which the intended behaviour is to round the numbers to fixed-precision floating-point numbers at every step of the loop.

\paragraph*{Organisation}
In Section~\ref{sec:prelim}, we formalise the 
model and problems and discuss some of the properties of floating-point 
rounding. In Section~\ref{sec:undec}, we present our undecidability result for the general case.
Finally, in Section~\ref{sec:posdec} we establish some special periodic structure associated with the orbit and use this structure in \cref{sec:decidablemodelchecking} to show that model checking is decidable for non-negative LDS.

\section{Preliminaries}
\label{sec:prelim}

\subsection{Linear dynamical systems and rounding functions}

\begin{definition}
A $d$-dimensional linear dynamical system (LDS) $(M,x)$ 
comprises a matrix $M\in\mathbb{Q}^{d\times d}$ and an initial 
vector $x\in\mathbb{Q}^d$. 

Given a rounding function $[\cdot]: \mathbb{Q}^d \to \mathbb{Q}^d$, 
and an LDS $(M,x)$ the rounded orbit $\mathcal{O}$  is the sequence $(x^{(t)})_{t\in \N}$ such that $x^{(0)}=[x]$ and  $x^{(t)} = [M x^{(t-1)}]$ for all $t\geq 1$. 
\end{definition}

Given $p\in \mathbb{N}$, we say that a number $x$ is a floating-point number with precision $p$ if $x =m \cdot 10^\alpha$ such that $m\in \mathbb{Q}$ is a decimal number in $\{0\}\cup[0.1,1)$ with $p$ digits in the fractional part (after the decimal point) and $\alpha \in \mathbb{Z}$. In particular, we associate by convention the number with mantissa $m=0$ to the exponent $-\infty$.
Given a number $x = m\cdot 10^\alpha$ we define $\mantissa(x) = m$ and $\exponent(x) = \alpha$.

We are interested in the floating-point rounding function $[\cdot]$ with precision $p\in\N$.
Given a real number $x\in \mathbb{R}$, we define $[x]$, the floating-point rounding of $x$, as the closest floating-point number with precision $p$ based on the first $p+1$ digits of $x$. 

Where there are two possible choices, any deterministic choice that is consistent with the properties listed below is acceptable.\footnote{For example, always rounding up, always rounding down, round to even, rounding towards zero, rounding away from zero are acceptable, providing the choice is fixed.}
We denote by $\mathbb{FP}_{10}[p]$ the subset of $\mathbb{Q}$ representable in base $10$ as a floating-point numbers with $p$ digits. 
We use the following useful properties of the rounding function:

\begin{itemize}
\item it is \emph{log-bounded}, \emph{i.e.}\ there exists a constant $c \in \mathbb{R}_+$ 
such that $\forall x \in \mathbb{R},\frac{|x|}{c} \leq |[x]|\leq c |x|.$
\item it is \emph{mantissa-based}, \emph{i.e.}\ if $x =10^\alpha x'$, then 
$[x] = 10^\alpha [x']$.
\item it is \emph{$(p+1)$-finite}, \emph{i.e.}\
the output of the rounding is not dependent on the $i$-th digit of the mantissa, for each integer $i > p+1$. In other words, if $x$ and $x'$ agree on the first $p+1$ digits 
then $[x] = [x']$.
\item it is \emph{sign preserving}, \emph{i.e.}\ $\operatorname{sign}(x) = \operatorname{sign}([x])$. The fact that $[x] = 0 $ if and only if $x = 0$ also follows from the log-bounded property.
\end{itemize}

The floating-point rounding is defined above on a single real. It is extended
straightforwardly to a vector $x$ by applying it to each of its components $(x)_i$ where
$i$ ranges from $1$ to the dimension of the vector. 
As such, the term $[Mx]$ is obtained by first computing exactly the the vector
$Mx$ and then by rounding each component $(Mx)_i$.
An alternative approach could be to maintain each sub-computation in $p$-bits
of precision, \emph{but this is not the approach we take}. Such an orbit can be simulated in our setting by increasing the dimension so that operations can be staggered in a way that at most one operation (scalar product or variable addition) is used in each assignment.

\subsection{Model checking}

We consider the model-checking problem of an LDS over semialgebraic sets.

\begin{definition}
A semialgebraic set $Y\subseteq \mathbb{R}^d$ is defined by a finite Boolean combination of polynomial inequalities.
\end{definition}

Let $(M,x)$ be an LDS with rounded orbit $\mathcal{O}$ and $\mathcal{\target} = \{\target_1,\dots, \target_k\}$ be a 
collection of semialgebraic sets. 
The characteristic word of $\mathcal{O}$ is $w =w_1w_2w_3\ldots\in (2^{\{1,\dots,k\}})^\omega$, such that 
$j \in w_t$ if and only if $x^{(t)} \in \target_j$.

The model-checking problem asks whether the characteristic word is contained within a given $\omega$-regular language, usually specified in a temporal logic such as monadic second order logic (MSO), or often its LTL fragment. Without loss of generality we assume that the property is given as a B\"uchi automaton~\cite{buchi1990decision}.

\begin{problem}[Floating-point Model-checking Problem]
\label{problem:modelchecking}
Given an LDS $(M,x)$ with rounded orbit $\mathcal{O}$,
a collection of semialgebraic sets $\mathcal{\target} = \{\target_1,\dots, \target_k\}$ and an $\omega$-regular specification $\phi$, the model-checking problem consists in deciding whether the characteristic word $w$ of $\mathcal{O}$  satisfies the specification $\phi$.
\end{problem}

We will also consider the point-to-point reachability problem, which is a subcase of the model-checking
problem (\cref{problem:modelchecking}):

\begin{problem}[Floating-point Point-to-point Reachability Problem]\label{problem:reach}
Given a $d$-dimensional LDS  $(M,x)$, and a target vector $y\in \mathbb{Q}^d$, the point-to-point reachability problem consists in deciding whether
$y$ belongs to the rounded orbit $\mathcal{O}$.
\end{problem}

Given a target $\target\subseteq \mathbb{R}^d$, we associate the set of hitting times $\mathcal{Z}(\target) = \{t \mid x^{(t)}\in \target\}$. Under this formulation, the reachability problem is reformulated as whether $\mathcal{Z}(\target)$ is empty. However, for model checking we will develop a more comprehensive understanding of the hitting times of each target $\target_1,\dots, \target_k$.

\subsection{Structure of \texorpdfstring{$M$}{M}}

Formally, $M$ is a $d$-dimensional matrix  indexed by the elements $\{1,\dots,d\}$.
However, we interpret $M$ as an automaton over states $Q =\{q_1,\dots, q_d\}$ and reference the entries of $M$ by pairs of states. That is, we refer to $M_{q_1,q_2}$ rather than $M_{1,2}$.

We denote by $G_M$ the weighted directed graph whose adjacency matrix
is $M$. That is, a graph with vertices $Q$ and with an edge from $q_j$ to $q_i$ weighted by $M_{q_i,q_j}$ if $M_{q_i,q_j} \ne 0$.\footnote{Note that the orientation of the edge may appear switched from the reader's expectation. This is due to the convention that $M$ is pre-multiplied with $x$ at every step.}

Let $S_1, \cdots, S_s\subseteq Q$ be the strongly connected components (SCCs) of $G_M$. Our analysis will consider each strongly connected component separately, thus it will often be useful to consider the entries of $x \in \mathbb{FP}_{10}[p]^{Q}$ corresponding only to one strongly connected component. 
Without loss of generality, by reordering the states where necessary, we assume that the states in $Q$ are ordered so that states within the same SCC appear next to one another, and the strongly connected components are topologically sorted, \emph{i.e.}\ there is no edge from 
$S_i$ to $S_j$ where $i > j$.
We split a vector $x$ into $s$ smaller vectors, denoted $x_{S_1},\dots, x_{S_s}$, each representing the entries of $x$ corresponding to the SCC\@. Letting $x_{S_j} = (z_{1, j}, \cdots, z_{d_j, j})^T$ and $|S_j| = d_j$, we thus have  $x$ is partitioned as
\[
x = (z_{1, 1} \cdots  z_{d_1, 1} , \cdots  , z_{1, s} \cdots z_{d_s, s} )^T.
\]

Moreover, for each pair of SCCs $S_i,S_j$, we denote by $M_{S_i,S_j}$ the submatrix of $M$ restricted to the rows related to $S_i$ and columns related to $S_j$, which is a matrix with $d_i$ rows and $d_j$ columns. If $S_i=S_j$, we simply write $M_{S_i}$. In other words, $M_{S_i,S_j}$ is the matrix that shows the dependency between $S_i$ and $S_j$, and we have
\[M = 
\begin{pmatrix}
M_{S_1} & M_{S_1, S_2} & \cdots & M_{S_1, S_s} \\
M_{S_2, S_1} & M_{S_2} & \cdots & M_{S_2, S_s} \\
\vdots & \vdots & \ddots & \vdots \\
M_{S_s, S_1} & M_{S_s, S_2} & \cdots & M_{S_s}
\end{pmatrix}
\]

We say $S_i$ \emph{feeds} $S_j$, and $S_j$ is \emph{fed by} $S_i$ if there is some edge in $G_M$ from some state in $S_i$ to some state in $S_j$.

\section{Undecidability of point-to-point reachability}
\label{sec:undec}

In this section, we give a sketch of the proof of the undecidability of
\cref{problem:reach} (and thus of \cref{problem:modelchecking}) in the general case. 
The full proof is postponed to \cref{sec:app_undec}.

\thmundec*

This result is obtained by reduction from the termination of a two-counter Minsky machine. 
We recall the definition of this model:
\begin{restatable}{definition}{definitionminski}
A two-counter Minsky machine is defined by a finite set of states
$\ell_1,\dots,\ell_m$, a distinguished starting state
(w.l.o.g. $\ell_1$), a distinguished halting state
(w.l.o.g. $\ell_m$), two natural integer counters, here denoted as $x$
and $y$, and a mapping deterministically associating to each state
transition a particular action.

Each transition takes one of the following forms: for $z\in \{x,y\}$,
\begin{description}
\item[increment] $\operatorname{inc}_z(\ell_j)$: add 1 to counter $z$, move to state $\ell_j$.
\item[decrement] $\operatorname{dec}_z(\ell_j)$: remove 1 from counter $z$ if $z>0$, move to state $\ell_j$.
\item[zero test] $\operatorname{zero?}_z(\ell_j,\ell_k)$: if $z=0$ move to state $\ell_j$ else move to state $\ell_k$.
\end{description}

The configuration of a two-counter Minsky machine consists of the current state and the values of $x$ and $y$. 
\end{restatable}

Without loss of generality (by first using a zero test), one can assume a decrementation 
operation is never used in a configuration where the counter to be decreased has
value $0$, hence removing the need to check whether $z>0$.
 
The halting problem asks whether, starting in configuration $(\ell_1,0,0)$, that is, in the distinguished starting state with both counters set to $0$, whether the state 
$\ell_m$ is reached. The problem is undecidable~\cite{minsky1967computation}.

We build an LDS with mantissa length $p=1$ and base $10$ that simulates a run of a given Minsky machine. 
The reduction happens to maintain the invariant that each mantissa always has the value $0$ or $1$ after 
rounding (although, as we operate in base 10, there are 10 possible values the mantissa could have taken).
For ease of readability, we describe this LDS using variables to represent the dimensions
and linear functions to represent the transition matrix.
For each state of the Minsky machine, we use two variables corresponding to the two counters. 
Throughout the simulation, if the Minsky machine is in state $j$, the counter values are stored in the exponents of the variables 
associated with state $j$, and all other variables are zero.

The crux of our reduction lies in the handling of the zero test. More precisely, 
suppose we need to branch depending on whether $x$ is equal to $0$, then we need to define linear 
transitions that transfer the values of the two counters from one pair of variables 
to the appropriate new pair of variables.
This is done using filter functions: 
the function $\operatorname{filter}_+(u,v)$ (resp. $\operatorname{filter}_-(u,v)$) is equal to $v$ if $v\geq u$ (resp. $v<u$) and to $0$ otherwise.
We end this sketch with the construction of these functions and proof
that they operate as advertised.

\begin{restatable}{lemma}{lemfilter}
\label{lem:filter}
Given $u,v$ of the form $10^c$ with $c \in \N$,
one can compute the value $w=\operatorname{filter}_+(u,v)$ in three linear operations
with floating-point rounding.
\end{restatable}
\begin{proof}
We compute $w=\operatorname{filter}_+(u,v)$ in three successive operations 
using two temporary variables, $temp$ and $temp2$, initially set at $0$ (recall, rounding is applied after each step):\\
\begin{tabular}{p{1cm}ll}
&$temp$ & $\leftarrow u + v$\\
&$temp2$ & $\leftarrow temp - u$\\
&$w$ & $\leftarrow 1.1 * temp2$
\end{tabular}

\noindent
Let $c_1,c_2\in \N$ such that $u=10^{c_1}$ and $v=10^{c_2}$. Recall that the notation $[\cdot]$ is the floating-point rounding function.

\noindent First observe that if $c_1=c_2$:\\
\begin{tabular}{p{1cm}ll}
&$temp$ & $\leftarrow [10^{c_1} + 10^{c_2}] = 2\cdot 10^{c_1}$\\
&$temp2$ & $\leftarrow [2\cdot10^{c_1} - 10^{c_1}] = 10^{c_1} (= v)$\\
&$w$ & $\leftarrow [1.1\cdot10^{c_1}] = 10^{c_1} = v\quad \text{ as required.}$
\end{tabular}

\noindent Secondly, assume that $u > v$, and thus $c_1> c_2$:\\
\begin{tabular}{p{1cm}ll}
&$temp$ & $\leftarrow [10^{c_1} + 10^{c_2}] =  10^{c_1} = u$\\
&$temp2$ & $\leftarrow [ 10^{c_1} - 10^{c_1}] = 0$\\
&$w$ & $\leftarrow [1.1\cdot 0 ] =0\quad \text{ as required.}$
\end{tabular}

\noindent We split the case that $v > u$, thus $c_2 > c_1$, into two cases. Suppose $c_2 > c_1 + 1$:\\
\begin{tabular}{p{1cm}ll}
&$temp$ & $\leftarrow [10^{c_1} + 10^{c_2}] =  10^{c_2} = v$\\
&$temp2$ & $\leftarrow [ 10^{c_2} - 10^{c_1}] = [0.\underbrace{99\dots99}_{c_2-c_1\ge 2} \cdot 10^{c_2}] = 1\cdot 10^{c_2} = v$\\
&$w$ & $\leftarrow [1.1\cdot 10^{c_2} ] =10^{c_2} = v\quad \text{ as required.}$\\
\end{tabular}

\noindent Finally, $c_2 = c_1 +1$:\\
\begin{tabular}{p{1cm}ll}
&$temp$ & $\leftarrow [10^{c_1} + 10^{c_2}] =  10^{c_2} = v$\\
&$temp2$ & $\leftarrow [ 10^{c_2} - 10^{c_1}] = [0.9 \cdot 10^{c_2}] = 9\cdot 10^{c_2-1}$\\
&$w$ & $\leftarrow [1.1\cdot 9\cdot 10^{c_2-1} ] =[9.9\cdot 10^{c_2-1}] = 10\cdot10^{c_2-1}= 10^{c_2} = v\qquad\quad$\\
&& \hfill$\text{ as required.}$\qed
\end{tabular}
\end{proof}

\begin{restatable}{corollary}{corfilter}
\label{cor:filter}
Given $u,v$ of the form $10^c$ with $c \in \N$, 
one can compute the value $w=\operatorname{filter}_-(u,v)$ in four linear operations
with floating-point rounding.
\end{restatable}
\begin{proof}
Observe that $\operatorname{filter}_-(u,v) = v- \operatorname{filter}_+(u,v)$, which can be encoded in four steps by first computing $\operatorname{filter}_+(u,v)$ in three steps.
\qed\end{proof}

\section{Pseudo-periodic orbits of non-negative LDS}
\label{sec:posdec}

We shift our focus to proving that model checking is decidable for systems with non-negative matrices. We first establish the behaviour of the system in this section
and then complete the proof of \cref{thm:modelcheckingdec} in \cref{sec:decidablemodelchecking}. Our main result is that the rounded orbit of an LDS is periodic in the following sense, which we call \emph{pseudo-periodic}.

\begin{definition} 
A sequence $(x^{(t)})_{i \in \N}$ of $d$-dimensional vectors of floating-point numbers is 
called pseudo-periodic if and only if there exists a starting point $N
\in \mathbb{N}$, period $T \in \mathbb{N}$ and growth rates
$\alpha_1,\dots,\alpha_d \in \mathbb{Z}$ such that
\[\forall t \geq N, \forall j\in\{1,\dots,d\}, 
(x^{(t + T)})_j = 10^{\alpha_j} (x^{(t)})_j.\]

We say the sequence is \emph{effectively} pseudo-periodic if the defining constants $N,T,\alpha_1,\dots, \alpha_d$ can be computed.
\end{definition}

\begin{theorem}\label{thm:main:pseudo-periodic}

Let $(M,x)$ be a $d$-dimensional LDS  
where $M$ is non-negative
and let $(x^{(t)})_{t\in \N}$ be its rounded orbit.

The rounded orbit $(x^{(t)})_{t\in \mathbb{N}}$ is effectively pseudo-periodic.
\end{theorem}

In order to establish this result, we will find some partitions of the graph associated 
to $M$ such that each part is effectively pseudo-periodic with the
same increasing rate $\alpha$ for every state in the partition.

\subsection{Preprocessing periodicity}
The core of our approach is to show that, within each SCC of the graph associated to $M$, the values associated with states are of similar magnitude. 
This is however only true if the SCC is aperiodic.
When a state is in a periodic SCC its value could change drastically depending on which phase the system is in. For example, consider a simple alternation between two states, in which the value is very large in one state and very small in the other; the states will alternate between big and small values.

We ``hide'' these periodic behaviours by blowing up the system so that 
each SCC of the new system describes only one of the periodic 
subsequence and we will subsequently show that the value of each state in an SCC is either zero or of a similar magnitude.

We apply the following construction to our system. Let $\period$ be the period, defined as the least common multiple of the length of every simple cycle in the graph.
Let $Q$ be the indices of $M$ (\emph{i.e.}\ the states of the
generated automaton).
We define new states $Q' = Q\times \{0,\dots,\period-1\}$ by annotating each state in $Q$ with the phase. To avoid cluttering notation we will regularly refer to states in $Q'$ in the form $(q,i+\ell)$ for $\ell\in\mathbb{Z}$, on the understanding that the phase, $i+\ell$, is normalised into $ \{0,\dots,\period-1\}$ by taking the residue modulo $P$ if necessary.
We define a new matrix $M'$ over the states $Q'$ such that $M'_{(q,i+1),(q',i)} = M_{q,q'}$ for $i\in\{0,\dots,\period-1\}$, and zero otherwise. 
We initialise a new starting vector $x^{(0)}_{(q,0)} = x^{(0)}_q$ and $x^{(0)}_{(q,i)} = 0$ for $i\in\{1,\dots,\period-1\}$.

Intuitively, at each time step $t$ 
the vector generated by the original system is equal to the vector of the new system
restricted to the states indexed by $i \equiv t\mod \period$ and every state with another index
is equal to $0$.

Let $S \subseteq Q$ be a strongly connected component. In $Q'$ there exists strongly connected components $S'_1,\dots, S'_k\subseteq Q'$ with $k\le |S|$ such that 
$\bigcup_{i= 1}^k S'_i = S \times \{0,\dots,\period-1\}$.
Each set $S'_j$ is periodic, with period $\period$.

Henceforth in the rest of this section we work on the system $( M', x')$ implicitly over states $Q'$ which, by overloading of notation, we rename $( M, x)$ over $Q$ to avoid cluttering notation.

Note that this transformation also requires to marginally complicate the targets.
Indeed, consider a set $\target\subseteq \mathbb{R}^Q$.
We define the sets $\target/i$ for $i<\period$ such that 
$\target/i = \{y \in \mathbb{R}^{Q'}\mid \exists y'\in \target\ :\ y_{(q,i)}= y'_q \text{ for } q\in Q \text{ and }y'_{(q,j)} = 0 \text{ for } j\ne i \}$. 
The hitting times of $\target$, $\mathcal{Z}(\target)$, in the original LDS can then be 
obtained in the new LDS as the disjoin union: $\bigcup_{i\in\{0,\dots,\period-1\}} \mathcal{Z}(\target/i)$. It suffices to characterise the hitting times for each $\target/i$.

\subsection{Pseudo-periodicity within top SCCs}

Let us first consider top SCCs, these are SCCs with no incoming edges from states of other SCC, and therefore the value of each variable at each step depends only on the value of 
states in the same SCC.

\begin{lemma}
\label{lem:wietland}
Let $S_j$ be a strongly connected component of $(M, x)$.
Let $S_{j,i} = \{(q,i) \in S_j\}$ be the states associated with $S_j$ from the $i$-th phase.

There exists $C\leq \period d^2$, such that, for every $i,j$, 
$(M^C)_{S_{j,i}}$ is positive.%
\end{lemma}
\begin{proof}

The matrix $(M^\period)_{S_{j,i}}$ is non-negative, irreducible (\emph{i.e.}, its graph is 
strongly connected) and of period 1.
As such, $(M^\period)_{S_{j,i}}$ is primitive~\cite{Boyle2005NOTESOT}
which means that a power $C'$ of this matrix is positive. 
The theorem follows with $C=\period C'$.
Moreover, $C'$ is at most $d^2-2d+2$~\cite{SCHNEIDER20025}.
\qed \end{proof}

Our goal is to show that within an SCC, each of the non-zero entries are of a similar magnitude due to the presence of a relatively short path ($C$) between any two states in the SCC. To do this we introduce the notion of closeness and observe some useful properties.

\begin{definition}
We say two numbers $x,x'\in\mathbb{FP}_{10}[p]$ are $\delta$-close, denoted  by $x\approx_\delta x'$ if $\abs{\exponent(x) - \exponent(x')} < \delta$. 
In particular, for every $\delta> 0$, zero is assumed to be $\delta$-close only to itself.

We extend the notion to vectors $y,y\in \mathbb{FP}_{10}[p]^S$, indexed by $S\subseteq Q$, such that $y\approx_\delta y'$ if all entries of the same phase are $\delta$-close to one another across both $y$ and $y'$, that is, for each phase $i\in\{0,\dots,P-1\}$ and all
$(q,i),(q',i)\in S$: $y_{(q,i)}\approx_\delta y'_{(q',i)}$, $y_{(q,i)}\approx_\delta y_{(q',i)}$ and $y'_{(q,i)}\approx_\delta y'_{(q',i)}$.
\end{definition}

\begin{restatable}{proposition}{pcloseness}\label{prop:closeness}
Let $x,x'\in\mathbb{FP}_{10}[p]$ be non-zero floating-point numbers.
\begin{enumerate}[label=\textnormal{(\arabic*)}]
\item If $x\approx_\delta x'$ then $10^{-\delta-1} \le x/x' \le 10^{\delta+1}$.
\item If $10^{-\delta} \le x/x' \le 10^{\delta}$ then $x\approx_{\delta+2} x'$. 
\item If $x\approx_{\delta} x'$ and $x' \approx_{\eta} x''$ then $x\approx_{\delta+\eta+4} x''$. 
\end{enumerate}
\end{restatable}

\begin{lemma}
\label{lem:closetop}
Let $S_j$ be a top strongly connected component of $(M, x)$, 
and let $C$ be as given by Lemma~\ref{lem:wietland}.

There exists $\beta\in \mathbb{N}$ such that for all $(q,i), (q',i)\in S_{j}$ and every $t \ge C$ then
\begin{itemize}
    \item if $t\not\equiv i\mod \period$, then $x^{(t)}_{(q,i)} = 0$, 
    \item otherwise, $x^{(t)}_{(q,i)} \approx_{\beta} x^{(t)}_{(q',i)}$.
\end{itemize}
\end{lemma}

\begin{proof}
Let $t\in\N$.
If $t\not\equiv i \mod \period$ then $x^{(t)}_{(q,i)} = 0$ for all $(q,i)\in S_{j,i}$ by construction.

Otherwise, let $\displaystyle m \ge \max_{q,q'\in Q : M_{q,q'}\neq 0}\max\left(M_{q,q'}, (M_{q,q'})^{-1}\right)$ be a constant larger than all values occurring in $M$ and so that $\frac{1}{m}$ is smaller than all non-zero values appearing in $M$. Let $c$ be the constant from the log bounded property of the rounding function $\round{\cdot}$ and $d$ be the dimension of $M$.

Observe that for all $t\in \N$ with $t=i\mod \period$ we have \begin{align*}
x^{(t)}_{(q,i)} &= \round{\sum_{(q',i-1)} M_{(q,i),(q',i-1)}x^{(t-1)}_{(q',i-1)}}
\\ & \ge\frac{1}{c}{\sum_{(q',i-1)} M_{(q,i),(q',i-1)}x^{(t-1)}_{(q',i-1)}}\tag{by log bounded}
\\ & \ge \frac{1}{cm} \quad  \max_{(q',i-1) \text{ s.t. }  M_{(q,i),(q',i-1)}> 0}{x^{(t-1)}_{(q',i-1)}} \tag{by defn of $m$}
\\ \text{In particular}
\\ x^{(t)}_{(q,i)} &\ge
 \frac{1}{cm} {x^{(t-1)}_{(q',i-1)}} \text{ for all } (q',i-1) \st M_{(q,i),(q',i-1)}> 0
\end{align*}

Using induction we obtain:  \[x^{(t+k)}_{(q,i+k)} \ge \frac{1}{(cm)^{k-1}}x^{(t+1)}_{(q',i+1)} \ge \frac{1}{(cm)^k} {x^{(t)}_{(q'',i)}}
\]
for all $(q',i+1), (q'',i)$ such that $  M^{k-1}_{(q,i+k),(q',i+1)} > 0$ and $M^{}_{(q',i+1),(q'',i)} >0$.

In particular, we have $x^{(t+C)}_{(q,i)} \ge \frac{1}{(cm)^{C}}x^{(t)}_{(q',i)} $  for all $q'$ (since $M^{C}_{(q,i),(q',i)} >0$ for all $q'$ by the previous lemma). 

On the other hand we have  \[x^{(t+1)}_{(q,i+1)} = \round{\sum_{q': M_{(q,i+1),(q',i)} >0} M_{(q,i+1),(q',i)}x^{(t)}_{(q',i)}} \le m c d\max_{(q',i)\in S_j}x^{(t)}_{(q',i)}.\] 
By induction we get that $x^{(t+C)}_{(q,i)} \le {(mcd)^{C}}\max_{(q',i)\in S_j}x^{(t)}_{(q',i)} $. 
Hence, for all $q,q'\in S_j$ we have 
\begin{equation*}\frac{1}{(mc)^{C}} \max_{(q'',i)\in S_j} x^{(t)}_{(q'',i)} \le x^{(t+C)}_{(q',i)}\quad \text{ and }\quad x^{(t+C)}_{(q,i)} \le {(mcd)^{C}}\max_{(q'',i)\in S_j} x^{(t)}_{(q'',i)}.\end{equation*}
Hence $\frac{x^{(t+C)}_{(q,i)}}{x^{(t+C)}_{(q',i)}} \le d^C(mc)^{2C}$.

Setting $\gamma=\ceil{\log_{10} d^C(mc)^{2C} }$, we thus have that 
$10^{-\gamma}x^{(t+C)}_{(q',i)} \le x^{(t+C)}_{(q,i)} \le 10^\gamma x^{(t+C)}_{(q',i)}$ for all $(q,i),(q',i) \in S_{j,i}$ and $t\in \N$. Then $x^{(t)}_{(q',i)} $ and $ x^{(t)}_{(q,i)}$ are $\beta = \gamma+2$ close by \cref{prop:closeness}.  
\qed \end{proof}

\begin{lemma}
\label{lem:effective-top}
Let $S_j$ be a top strongly connected component of $(M, x)$. Then the sequence
$(x^{(t)}_{S_j})_{t\in \N}$ is effectively pseudo-periodic.
\end{lemma}
\begin{proof}
Let $\beta$ and $C$ be as in Lemma~\ref{lem:closetop}.
Denote $q_1,\dots, q_m$ the states of $S_j$.
We define the sequence $(y^{(t)})_{t\geq C}$ such that for all $t\geq C$ and $q\in S_j$
denoting $ (p^{(t)})_{q} = \mantissa([x^{(t)}_{q}])$ and 
$(\alpha^{(t)})_{q}=\exponent([x^{(t)}_{q}])$ we have that
$y^{(t)}= (p_{q_1},0, p_{q_2},\alpha_{q_2}-\alpha_{q_1},\dots,p_{q_m},
\alpha_{q_m}-\alpha_{q_1}) $. Note that this sequence can only take finitely many values 
as the mantissas have a precision of $p$ decimals and by Lemma~\ref{lem:closetop}, 
for all $k\leq m$, $\alpha_{q_k}-\alpha_{q_1}\in \{-\beta,\dots, \beta\}$.
As a consequence, the sequence $(y^{(t)})_{t\geq C}$ takes the same value multiple times.
Let $k_1$ and $k_2$ be the two distinct minimal integers such that $y^{(k_1)}=y^{(k_2)}$. 
Setting $\alpha=\alpha^{(k_2)}_{q_1}-\alpha^{(k_1)}_{q_1}$
We have that $x^{(k_1)} = x^{(k_2)}\cdot 10^{\alpha}$.
Since $\round{\cdot}$ is mantissa-based, one can show by induction that for all $t\geq 0$,
$x^{(k_1+t)} = x^{(k_2+t)}\cdot 10^{\alpha}$.
Therefore the sequence $(x^{(t)}_{S_j})_{t\in \N}$ is effectively pseudo-periodic with 
period $T=k_2-k_1$ and starting point $N=C+k_1$.

Moreover, as the maximum number of different values taken by $(y^{(t)})_{t\geq C}$ is
known, we can deduce that both $k_1$ and $k_2-k_1$ are smaller than 
$10^{p m}(2\beta + 1)^{m}+1$.
\qed \end{proof}
Note that the increasing rate is the same for every state of the strongly connected component.

\subsection{Pseudo-periodicity within lower SCCs}

We consider a strongly connected component $\Sme$, which is fed by at least one strongly connected components $F_1,\dots,F_\ell$, $\ell \ge 1$. We let $\Sfeeder = F_1\cup \dots \cup F_\ell$ and assume every $F_i$ is pseudo-periodic.

In this section we show 
\begin{theorem}
\label{thm:lowerSCC}
$\Sme$ is effectively pseudo-periodic and the growth rate of $\Sme$ is the same for all $q\in \Sme$.
\end{theorem}

We first observe that the difference between values in $\Sme$ is bounded.
This is achieved with a proof similar to the one of \cref{lem:wietland} and
\cref{lem:closetop} (though having to combine considerations of $\Sme$ and $\Sfeeder$).

\begin{restatable}{lemma}{zclosenesssme}
\label{lemma:zclosenessofsme}

There exists $\eta,N'\in \mathbb{N}$, such that for all $(q,i),(q',i)\in \Sme$, 
all $t \ge N'$ and all $i \in \{0,\dots,P-1\}$ then
\begin{itemize}
    \item if $t\not\equiv i\mod \period$, then $x^{(t)}_{(q,i)} = 0$,
    \item otherwise, $x^{(t)}_{(q,i)} \approx_{\eta} x^{(t)}_{(q',i)}$.
\end{itemize}

\end{restatable}

\begin{definition}

We say that $x_{q}^{(t)}$ is influenced by $\Sfeeder$ if  \[x^{(t)}_{q} = \round{\sum_{q'\in \Sfeeder} M_{q,q'}x^{(t-1)}_{q'}  + \sum_{q'\in \Sme} M_{q,q'}x^{(t-1)}_{q'}} \ne \round{\sum_{q'\in \Sme} M_{q,q'}x^{(t-1)}_{q'}} \]
and in particular $x_{q}^{(t)}$ is influenced by $u\in \Sfeeder$ if:
\[\round{\sum_{q'\in \Sfeeder \cup \Sme} M_{q,q'}x^{(t-1)}_{q'}  }\ne \round{\sum_{q'\in \Sfeeder \cup \Sme \setminus \{u\}} M_{q,q'}x^{(t-1)}_{q'}  }  .\] 

\end{definition}

We can restrict $\Sfeeder$ to the $F_i$ in $\Sfeeder$ with the maximum growth rate. 
Indeed, from some point on, any $F_i$ with non-maximal growth rate is much smaller than the 
maximal ones, and as by the proof of~\cref{lemma:zclosenessofsme} the values within $\Sme$ 
are close to (or greater than) the maximum value within $\Sfeeder$, 
this $F_i$ would not influence with any $x_{q}^{(t)}$ with $q\in \Sme$.
Let $N_1$ be the point from which we can assume, that the elements of $\Sfeeder$ are much larger than any other feeding SCCs and are thus the only ones potentially influencing of $\Sme$.

Since each $F_i$ is assumed to be pseudo-periodic, we have that $\Sfeeder$ pseudo-periodic. Let $T$ be the period of $\Sfeeder$, $N_2$ be the starting point and $\alpha$ be the growth rate of every state of $\Sfeeder$ (meaning the exponent of every state changes by $\alpha$ every $T$ starting form the $N$-th step.) Let $N = \max\{N_1, N_2\}$, that is, the point from which we can assume $\Sfeeder$ is both pseudo-periodic and dominating 
non-maximal SCCs feeding $\Sme$.

As a direct consequence of having the same growth rate, the non-zero terms within $\Sfeeder$ are close:

\begin{restatable}{proposition}{ppisclose}
\label{prop:ppiscloseness}
If a sequence of non-zero floating-point vectors $(v^{(t)})_{t\in \mathbb{N}}$ is 
pseudo-periodic with the same growth rate within a set $Q$, then there exists $\delta$ such that 
for all $q,q'\in Q$ and all $t\ge N$,
$v_{q}^{(t)} \approx_\delta v_{q'}^{(t)}$.
\end{restatable}

Moreover, either $\Sfeeder$ does not influence $\Sme$, or they are close.

\begin{restatable}{lemma}{lemmaB}
\label{lemma:B}
There exists $\beta,N\in \mathbb{N}$ such that:\\
For $t\geq N$ and $(q,i)\in\Sme$, 
if $x_{(q,i)}^{(t)}$ is influenced by $(q',i-1)\in \Sfeeder$,
then $x_{(r,i)}^{(t)} \approx_\beta x_{(r',i)}^{(t)}$ for all $(r,i),(r',i)\in \Sme\cup \Sfeeder$.
\end{restatable}

We will show \cref{thm:lowerSCC} through the following observation:
\begin{observation}
Observe that $\Sfeeder$ either influences $\Sme$ infinitely many times or finitely many times. We have two cases:
\label{obs:bothpp}
\begin{itemize}
    \item If $\Sfeeder$ influences $\Sme$ infinitely often, then they are infinitely often $\beta$-close by \cref{lemma:B}. Then we will observe through a simultaneous version of \cref{lem:effective-top} that $\Sme$ is pseudo-periodic.
    \item If $\Sfeeder$ influences $\Sme$ only finitely often, then clearly from some point on $\Sme$ behaves like a top SCC, and thus is pseudo-periodic directly by \cref{lem:effective-top}. 
\end{itemize}
\end{observation}
It will then remain to show that we can detect which of the two cases applies, and place a bound on the time to detect this, which will effectively reveal the constants of the pseudo-periodic behaviour. 

We now present a version of \cref{lem:effective-top}
to observe that if $\Sfeeder$ and $\Sme$ are infinitely often $\beta$-close then $\Sme$ is pseudo-periodic:

\begin{lemma}
\label{lemma:infcase}
Suppose $x_{\Sfeeder}^{(t)} \approx_\beta x_{\Sme}^{(t)}$ for infinitely many $t$.  Then there exists $t_1<t_2$, such that $x_{\Sfeeder}^{(t_1)} \approx_\beta x_{\Sme}^{(t_1)}$ and $x_{\Sfeeder}^{(t_2)} \approx_\beta x_{\Sme}^{(t_2)}$, $x_{\Sfeeder}^{(t_2)} = 10^\gamma x_{\Sfeeder}^{(t_1)}$ and $x_{\Sme}^{(t_2)} = 10^\gamma x_{\Sme}^{(t_1)}$. In particular, the sequence $(x_{\Sme}^{(t)})_{t\in \N}$ is pseudo-periodic with period $(t_2-t_1)$, starting 
from $t_1$ with growth rate of $\gamma$ in every state.
\end{lemma}
\begin{proof} At a time $t$ such that  $x_{\Sfeeder}^{(t)} \approx_\beta x_{\Sme}^{(t)}$,  we denote the vectors $x^{(t)}_{\Sfeeder} \in \mathbb{FP}_{10}[p]^{|\Sfeeder|}$ and $x^{(t)}_{\Sme} \in \mathbb{FP}_{10}[p]^{|\Sme|}$  respectively 
\begin{multline*}
(m_1^{(t)} 10^{\gamma^{(t)}_1},m_2^{(t)} 10^{\gamma^{(t)} +\alpha^{(t)}_2},\dots, m^{(t)}_{|\Sfeeder|} 10^{\gamma^{(t)} +\alpha^{(t)}_{|\Sfeeder|}}) \text{ and }\\ 
(n^{(t)}_1 10^{\gamma^{(t)} +\zeta^{(t)}_1},\dots, n^{(t)}_{|\Sme|} 10^{\gamma^{(t)} +\zeta^{(t)}_{|\Sme|}}),\end{multline*} where $m_i,n_i$ are taken from the finite set of mantissa values expressible in $p$ bits,  $\gamma^{(t)} \in \mathbb{Z}$ and $\alpha_i,\zeta_i \in \mathbb{Z}\cap[-\beta,\beta]$ denote the offset from $\gamma^{(t)}$.

Let $F$ bound the number of possible values $m_i,n_i,\alpha_i,\zeta_i$ can take on, where $F\le 10^{p(\abs{\Sfeeder}+\abs{\Sme})}\cdot (2\beta+1)^{\abs{\Sfeeder}+\abs{\Sme}-1}$. 
By the pigeonhole principle, after at most $F+1$ times in which  $x_{\Sfeeder}^{(t)} \approx_\beta x_{\Sme}^{(t)}$ there must exist two times $t_1 < t_2$ where the values of $m_i,n_i,\alpha_i,\beta_i$'s are equal (although the value of $\gamma$ could be different), thus $x_{\Sfeeder \cup \Sme}^{(t_2)} = \frac{10^{\gamma^{(t_2)}}}{10^{\gamma^{(t_1)}}}x_{\Sfeeder \cup \Sme}^{(t_1)}$.

Since the rounding function is mantissa-based, the system evolution from $x^{(t_1)}$ is equivalent to the systems evolution from $x^{(t_2)}= 10^\gamma x^{(t_1)}$, where $\gamma$ is the growth rate, $\gamma^{(t_2)}- \gamma^{(t_1)}$.
\qed \end{proof}

We can in fact decide whether $x_{\Sfeeder}^{(t)} \approx_\beta x_{\Sme}^{(t)}$ for the last time:
\begin{restatable}{lemma}{aretheycloseagain}
\label{lemma:finitecase}
Let $\beta,N$ be defined as in \cref{lemma:B}.
If $t\ge N$  then it is decidable whether there exists $t' > t$ such that $x_{\Sfeeder}^{(t')} \approx_\beta x_{\Sme}^{(t')}$.
\end{restatable}
\begin{proofsketch}[Full proof in \cref{proofof:lemma:finitecase}]
If we considered $\Sme$ in isolation, without the effect of $\Sfeeder$, we know it would be pseudo-periodic. We can simulate one period of $\Sme$ with and without the effect of $\Sfeeder$ and determine if $\Sfeeder$ influences $\Sme$ within one period. If it does then they must be close at this point. If $\Sfeeder$ does not influence $\Sme$ we know that $\Sme$ will behave pseudo-periodically at least until $\Sfeeder$ is close to $\Sme$ again; having established a growth rate for $\Sme$, we can compare the growth rates of $\Sfeeder$ and $\Sme$ to see if $\Sme$ will ever be close to $\Sfeeder$ again in the future. \qed
\end{proofsketch}

Finally to conclude the proof of \cref{thm:lowerSCC}, we refine \cref{obs:bothpp} to show that the period is bounded and thus the growth rates are computable:
\begin{itemize}
\item either $\Sfeeder$ is $\beta$-close to $\Sme$ infinitely often, in particular if they become close $F+1$ times then by \cref{lemma:infcase} it is pseudo-periodic. 
\item  or the system is pseudo-periodic because it behaves like a top-SCC, in which \cref{lem:effective-top} gives effective computation of the constants.
\end{itemize}
Which of these occurs is determined
by at most $F+1$ applications of \cref{lemma:finitecase}.

\section{Decidability of model checking}
\label{sec:decidablemodelchecking}
In this section we use the results obtained in the previous section to show that model checking is decidable. We use pseudo-periodicity to show that the characteristic word is eventually periodic, a case for which model checking is decidable.

\thmmodelchecking*

Consider a semialgebraic target $\target$, which can be expressed as a Boolean combination of polynomial inequalities over variables representing the dimensions. That is $\target = \{(x_1,\dots,x_d) \mid  \bigwedge_{i} \bigvee_j P_{ij}(x_1,\dots,x_n) \triangleright_{ij} 0 \}$, where $\triangleright_{ij}\in \{\ge,>,=\}$.

Given a linear dynamical system $(M,x)$ defining the rounded orbit $(x^{(n)})_{n=1}^{\infty}$, recall that $\mathcal{Z}(\target) = \{n \mid x^{(n)} \in \target\}$ are the hitting times of $\target$.  We claim that this set is semi-linear (equivalently eventually periodic) for semialgebraic $Y$.

\begin{definition}
A 1-dimensional linear-set, defined by a base $b\in\mathbb{N}$ and period $p\in\mathbb{N}$, is the set $\{x \mid \exists k \in \mathbb{N} : x = b + k\cdot p\}$. 
A semi-linear set is the finite union of a finite set $F\subseteq \mathbb{N}$ and linear sets. It can be assumed that each linear-set has the same period. Hence a 1-dimensional semi-linear set $X$ is defined by a finite set $F\subseteq \mathbb{N}$ and integers $m,p,b_1,\dots,b_m\in \mathbb{N}$ such that $x\in X$ if and only if $x\in F$ or $x = b + k\cdot p$ for some $k\in\mathbb{N}$ and $b\in\{b_1,\dots,b_m\}$. 
\end{definition}

\begin{theorem}\label{thm:semialgebraicissemilinear}
Let $\target$ be a semialgebraic target, $\mathcal{Z}(\target)$ is a semi-linear set.
\end{theorem}
\cref{thm:semialgebraicissemilinear} essentially completes the proof of \cref{thm:modelcheckingdec}. It is almost immediate that the characteristic word is eventually periodic (see \cref{lemma:eventuallyperiodic} in the appendix for a formal proof) and thus the model-checking problem can be decided by checking $A\cap\overline{B} = \emptyset$, where $A$ is an automaton representing the characteristic word and $B$ encodes the language of $\phi$.

It is standard that semi-linear sets are closed under intersection, union, and complementation (see~\cite{Haase18} for a nice introduction to semi-linear sets). Thus in order to express the hitting times of $\mathcal{Z}(\target)$ it is sufficient to express the hitting times of $ \{(x_1,\dots,x_d) \mid  P(x_1,\dots,x_n) \geq 0 \}$ for a finitely many polynomials $P$. Conjunction is found by taking the intersection of the hitting times, and disjunction by taking union. The hitting times of $P(x_1,\dots,x_n)>0$ can be rewritten as the complement of the hitting times of  $-P(x_1,\dots,x_n)\ge 0$. The hitting times of $P(x_1,\dots,x_n)=0$ is the conjunction (intersection) of $P(x_1,\dots,x_n)\ge0$ and $-P(x_1,\dots,x_n)\ge0$. Thus \cref{thm:semialgebraicissemilinear} is a consequence of the following lemma. 

\begin{lemma}Assume $x^{(t)} = (z_1^{(t)},\dots,z_d^{(t)} )_{i = 1}^{\infty}$, is a pseudo-periodic sequence with start point $N$, period $T$ and growth rates $\alpha_1,\dots,\alpha_n$ and $P \in \mathbb{Q}[x_1, \cdots, x_d]$ a rational polynomial in $d$ variables.\footnote{Some variables may be redundant, that is, if the polynomial does not depend on all dimensions of $x^{(t)}$ then some of the variables may not appear in $P$.} Then, $\{i \in \mathbb{N} \mid P(z_1^{(t)}, \cdots, z_{d}^{(t)}) \geq 0 \}$ is a semi-linear set.
\end{lemma}
\begin{proof}
First, we show that pseudo-periodicity is closed under product. Suppose $x_i^{(N+Tn)} = m_i 10^{\beta_i +\alpha_i\cdot n}$ and  $x_j^{(N+Tn)} = m_j 10^{\beta_j +\alpha_j\cdot n}$. Observe that  $x_i^{(N+Tn) } \cdot x_j^{(N+Tn) } = m_i\cdot 10^{\beta_i + \alpha_i n}m_j\cdot 10^{\beta_j + \alpha_jn} = m_im_j \cdot 10^{\beta_i + \beta_j + n(\alpha_i + \alpha_j)}$.  We conclude that the vector $(x_i \cdot x_j)^{(t)}$ is pseudo-periodic with growth rate $\alpha_i + \alpha_j$. Observe that the mantissa precision increase by at most 2.

Secondly, we show that if two pseudo-periodic sequences have the same growth rate, then their sum is also pseudo-periodic with the same growth rate. Suppose $x_i^{(N+Tn)} = m_i 10^{\beta_i +\alpha\cdot n}$, and $x_j^{(N+Tn)} = m_j 10^{\beta_j +\alpha\cdot n}$. Observe that $(x_i + x_j)^{(N+Tn)} = m_i 10^{\beta_i +\alpha\cdot n}+m_j 10^{\beta_j +\alpha\cdot n}= (m_i + m_j\cdot 10^{\beta_j-\beta_i}) 10^{\beta_i +\alpha\cdot n} $. Observe that the mantissa precision increased by at most $10^{|\beta_j-\beta_i|}$.

Let $P(x_1,\dots, x_n) = \sum_{i=1}^N c_i Z_i$, where $Z_i$ is a product of $x_1,\dots, x_n$.  Consider each monomial $Z_i$ occurring in $P$, since produce preserves pseudo-periodicity, we conclude that $Z_i$ is pseudo-periodic. $P^{(t)}$ is thus a linear combination of these pseudo-periodic vectors. Note our prior observation does not immediately imply that $P^{(t)}$ is pseudo-periodic as we required taking the sum of elements with the same growth rate. However, from some point on, we are only interested in those with the maximal growth rate.

Without loss of generality, let $Z_1,\dots, Z_r$ have the maximum-growth rate, and $Z_{r+1},\dots, Z_N$ have strictly smaller growth rate. For every $L \in \N$ there exists $N\in\mathbb{N}$ such that for all $t > N$,  $\exponent(Z_{1}^{(t)}) - \exponent(Z_{r+1}^{(t)}) > L$. 

Hence there exists $N\in\mathbb{N}$ such that for all $t > N$ if $\sum_{i=1}^r c_i Z_i > 0 $ if and only if $\sum_{i=1}^N c_i Z_i = \sum_{i=1}^r c_i Z_i + \sum_{i=r+1}^N c_i Z_i > 0$ because $\abs{\sum_{i=r+1}^N c_i Z_i} < \abs{\sum_{i=1}^r c_i Z_i}$ from some point on. Hence $\operatorname{sign}(\sum_{i=1}^{N} c_i Z_i^{(t)}) = \operatorname{sign}(\sum_{i=1}^{r} c_i Z_i^{(t)})$. 

Thus we restrict our attention to $\sum_{i=1}^{r} c_i Z_i^{(t)}$. Since each of the $Z_i$ for $i\in\{1,\dots,r\}$ have the same growth rate, we know that $\sum_{i=1}^{r} c_i Z_i^{(t)}$ is pseudo-periodic. Since $\operatorname{sign}(\sum_{i=1}^{r} c_i Z_i^{(t)})$ does not depend on the exponent, only the periodic mantissa, we have that the sign is periodic. The hitting times for $t \le N$ can be determined exhaustively and included in the finite set of the semi-linear set.
\qed \end{proof}

\subsubsection{Acknowledgements} 
Partially funded by DFG grant 389792660 as part of TRR~248 -- CPEC, see \href{https://perspicuous-computing.science}{\texttt{perspicuous-computing.science}}.
Joël Ouaknine is also affiliated with Keble College, Oxford as \href{http://emmy.network/}{\texttt{emmy.network}} Fellow.  David Purser was partially supported by the ERC grant INFSYS, agreement no. 950398.

\bibliographystyle{splncs04}

\newpage
\appendix

\section{Undecidability of point-to-point reachability}
\label{sec:app_undec}
In this section we show that, in general, \cref{problem:reach} (and thus \cref{problem:modelchecking}) is undecidable. 

\thmundec*

We reduce the halting of a two-counter Minsky machine to the point-to-point reachability problem. We recall here the definition of this model:

\definitionminski*

Without loss of generality (by first using a zero test), one can assume a decrement 
operation is never used in a configuration where the would-be decreased counter has 
value $0$, hence removing the need to check whether $z>0$.

The halting problem asks whether, starting in configuration $(\ell_1,0,0)$, that is, in the distinguished starting state with both counters set to $0$, whether the state 
$\ell_m$ is reached. The problem is undecidable~\cite{minsky1967computation}.

We describe below the construction of an LDS with mantissa length $p=1$ and base $10$ that will simulate a run of this machine. In particular, our reduction will maintain that the mantissa always has the value $0$ or $1$ after rounding (although, as we operate in base 10, there are 10 possible values the mantissa could have taken)\footnote{Technically we define floating-point numbers with mantissa in $\{0\}\cup[0.1,1)$. For convenience, throughout this section we write $1\cdot 10^{c}$ (or simply $10^c$) instead of $0.1\cdot10^{c+1}$.}.
For ease of readability, we describe the LDS using variables to represent the dimensions
and linear functions to represent the transition matrix.

In order to describe our proof we introduce two filter functions which will help us encode control flow. 
The function $\operatorname{filter}_+(u,v)$ (resp. $\operatorname{filter}_-(u,v)$) is equal to $v$ if $v\geq u$ (resp. $v<u$) and to $0$ otherwise.
The two following results are shown in \cref{sec:undec}.

\lemfilter*

\corfilter*

\begin{remark}
The proof of~\cref{lem:filter} technically requires encoding $1.1$ in the matrix. The problem setting we consider does not specifically require the matrix to have the same precision as the program variables and so no special encoding is required. However, in case one wishes to impose such restriction, we observe that it is possible to encode the multiplication by $1.1$ using only  floating-point numbers with precision $1$ by splitting the computation into $1$ and $0.1$. To do this we introduce an additional program variable $temp3$ and one additional linear operation, that is, we let:

\begin{tabular}{p{1cm}ll}
&$temp$ & $\leftarrow u + v$\\
&$temp2$ & $\leftarrow temp - u$\\
&$temp3$ & $\leftarrow 1\cdot 10^{-1}temp2$\\
&$w$ & $\leftarrow temp2 + temp3$.
\end{tabular}

\end{remark}

We now show the encoding of the two-counter machine into a linear dynamical system with rounding using the defined filter functions, entailing \cref{thm:undec}. 

\begin{proof}[Proof of \cref{thm:undec}]

First, for each state $\ell_j$ of the Minsky machine, we build variables $x_j, y_j$, and $a_j$, which will have the following invariant property:
when the run of the Minsky machine reaches a configuration ($\ell_j,x,y)$, then in the
corresponding run of the LDS we will have that $x_j=10^x, y_j=10^{y}$ and $a_j=10^{x+y}$.
Moreover, for all $k\neq j$, $x_k=y_k=a_k=0$.

Assuming the variable $\alpha_k$ corresponds to dimension $k$ in
the LDS, setting $\alpha_i \leftarrow \sum_{k} a_k \cdot \alpha_k$ means
that $M_{i,k} = a_k$.
Any entries which are not specified are assumed to be zero. 
When describing update transitions, we also want to use the 
$\operatorname{filter}$ operations. Since these operations represent up to four 
linear steps, we create four copies of each variable as well as temporary variables
($temp$ and $temp2$) for each state of the Minsky machine.
These copies and temporary variables are used implicitly: we only describe here 
the updates of the primary variable; the updates of the secondary variables can be deduced 
from Lemma~\ref{lem:filter} and Corollary~\ref{cor:filter}. Moreover, if an update 
function takes fewer than four steps, we complete it with updates that pass the value to 
the next secondary variable in order for every update to take exactly four steps.

Let us now define the update functions of the primary variables of the system.

\begin{itemize}
\item if the transition in $\ell_i$ is an increment, $\operatorname{inc}_z(\ell_j)$, then we multiply by 10 
the variables to keep the invariant on the variables and move the values to the next instruction. For instance, if $x$ is increased, we set
$x_{j}\leftarrow  10 \cdot x_i$,
$y_{j}\leftarrow y_i$,
and $a_{j}\leftarrow 10  \cdot a_i$.

\item if the transition in  $\ell_i$ is a decrement,  $\operatorname{dec}_z(\ell_j)$
then conversely to incrementation,
we  divide  the variables by 10 to keep the invariant on the variables and move 
the values to the next instruction, we set
$x_{j}\leftarrow  10^{-1} \cdot x_i$,
$y_{j}\leftarrow y_i$,
and $a_{j}\leftarrow 10^{-1}  \cdot a_i$.

\item if the transition in  $\ell_i$ is a zero test of $z$, $\operatorname{zero?}_z(\ell_j,\ell_k)$, we 
need to copy the values of our variables only to the correct coordinate. For that, we need
a way to filter their values, depending on the zeroness of $z$.
Assume the test is on $x$ without loss of generality. We define the following operations:
\begin{itemize}
\item $x_j\leftarrow  \operatorname{filter}_-(10, x_i)$,
\item $x_k\leftarrow  \operatorname{filter}_+(10, x_i)$,
\item $y_j\leftarrow  \operatorname{filter}_+(a_i, y_i)$,
\item $y_k\leftarrow  \operatorname{filter}_-(a_i, y_i)$,
\item $a_j\leftarrow  \operatorname{filter}_-(10\cdot y_i, a_i)$
\item $a_k\leftarrow  \operatorname{filter}_+(10\cdot y_i, a_i)$.
\end{itemize}
One can check that the $j$ variables are assigned the values from the $i$ variables if and only if $x=0$ (and thus $x_i=10^0$ and $a_j = y_j$), and the $k$ variables are assigned the values in the opposite case.

\item if the state is $\ell_m$ we zero the system: the values of the counters received in $x_m, y_m$ and the test value $a_m$ are discarded on the next step, thus
making every value of the system equal to $0$.
\end{itemize}

We have that the Minsky machine terminates if and only if the LDS described by the above behaviour,
starting with $x_1=y_1=a_1=1$ (and everything else at $0$) eventually hits the zero
vector.

This equivalence is a direct result of the invariant kept within the construction on the 
variables $x_j$ and $y_j$. Indeed, if the Minsky machine terminates, then following the 
same path the LDS we constructed puts all the stored values in $a_m,x_m,y_m$ associated to the terminating state $\ell_m$, thus discarding them and reaching the zero vector in the next step. And reciprocally,
the zero vector can only be reached by discarding the stored variables thanks to 
a halting instruction, proving halting of the Minsky machine.

Hence, point-to-point reachability is undecidable for LDS under floating-point rounding.
\qed\end{proof}

This proof, and in particular the construction of the filter functions rely on the use of 
base $10$. It works as well for most other bases. 
A notable exception is base 2 as the gap between $2^c$ and $2^{c+1}$ is not large enough 
for the rounding functions to operate correctly. It's possible however to artificially widen 
 this gap by storing the value of the counter $x$ for instance 
as $2^{2x}$ instead of $2^x$ in the variables of the LDS. With the guarantee that the 
exponent is even, the filter functions, and the rest of the proof, work in base 2 as 
well.

\section{Proof of Proposition~\ref{prop:closeness}}%

\pcloseness*

\begin{proof}\hfill

\noindent (1) Let $x = q_110^\alpha, x' = q_210^\beta$, with $\abs{\alpha-\beta}\le \delta$ and 
$q_1$ and $q_2$ are non-zero mantissa with $p$ decimals.
Then, $10^{-p}\leq \frac{q_1}{q_2}\leq 10^p$.

Hence
$\frac{x}{x'}=\frac{q_110^\alpha}{q_2 10^\beta}\le \frac{q_110^{\beta+\delta}}{q_210^{\beta}}\le \frac{10^\delta}{ 0.1 \cdot} \le 10^{\delta + 1}$.

And $\frac{x}{x'}=\frac{q_110^\alpha}{q_210^\beta}\ge \frac{q_110^{\alpha}}{q_210^{\alpha-\delta}}\ge \frac{ 0.1 \cdot}{1\cdot10^\delta} \ge 10^{-\delta - 1}$.

\medskip

\noindent (2) Let $x = q_110^\alpha, x' = q_210^\beta$, with $  10^{-\delta} \le x/x' \le 10^\delta$. 

Then $\frac{0.1 \cdot 10^\alpha}{10^\beta}\le \frac{q_110^\alpha}{q_210^\beta}= \frac{x}{x'} \le 10^\delta$, so $10^{\alpha-\beta}\le 10^{\delta+1}<10^{\delta+2}$

and $\frac{ 10^\alpha}{0.1 \cdot10^\beta}\ge \frac{q_110^\alpha}{q_210^\beta}= \frac{x}{x'} \ge 10^{-\delta}$, so $10^{\alpha-\beta}\ge 10^{-\delta-1}>10^{-\delta-2}.$

\medskip

\noindent  (3)  If $x\approx_{\delta} x'$ and $x' \approx_{\eta} x''$ then 
$\frac{x}{x'}\le 10^{\delta+1}$ and $\frac{x'}{x''} \le 10^{\eta+1}$ by (1). Hence $\frac{x}{x''} \le 10^{\delta+\eta+2}$ and similarly $\frac{x''}{x} \le 10^{\delta+\eta+2}$. Hence by (2) we have $x \approx_{\delta+\eta+2+2} x''$.
\qed\end{proof}

\section{Proof of Lemma~\ref{lemma:zclosenessofsme}}

\zclosenesssme*
\begin{proof}
Let $C$ be such that $(M^C)_{\Sme,\Sfeeder}$ and $(M^C)_{\Sme}$ is positive 
(\emph{i.e.} there is a path in the graph associated to $M$ from each element of 
$\Sme\cup \Sfeeder$ to each element of $\Sme$).
This integer exists as, 
from \cref{lem:wietland}, there exists $C_i$ and $C_0$ such that $M^{C_i}_{F_i}$
and $M^{C_0}_{\Sme}$ are positive.
As the SCC $F_i$ feeds $\Sme$, the $M_{\Sme,F_i}$ are non-zero non-negative matrices, 
and in particular, there is a path of length $C_i+C_0+1$ (corresponding to the an element
of $M^{C_0}_{\Sme}M_{\Sme,F_i}M^{C_i}_{F_i}$) between any state of $F_i \cup \Sme$ to any 
state of $\Sme$.
Setting $C$ as the product of $C_0$ and of the $C_0 + C_i + 1$, we have that 
there is a path of length $C$ between any state of $\Sfeeder \cup \Sme$ to any 
state of $\Sme$.

Recall that $m$ is a constant larger than all entries of $M$ and $\frac{1}{m}$ is smaller than all non-zero entries of $M$. Recall $c$ is the constant such that $[\cdot]$ is log-bounded. Let $d$ be $|\Sfeeder \cup \Sme|$.

Let us bound the effect of $C$ steps, first from above, for all $u\in \Sme$, we have:
\begin{align*}
x_{(q,i)}^{(t+C)} &= \round{\sum_{r \in \Sfeeder \cup \Sme} M_{(q,i),(r,i-1)}x_{(r,i-1)}^{(t+C-1)} } 
\\&\le mcd \max_{(r,i-1): M_{(q,i),(r,i-1)} >0} x_{(r,i-1)}^{(t+C-1)}
\\ &\le (mcd)^2 \max_{(q,i-2) : M^2_{(q,i),(r,i-2)} >0} x_{(r,i-2)}^{(t+C-2)}
\\ &\le \ \dots\ \le (mcd)^{C} \max_{(r,i) : M^{C}_{(q,i),(r,i)} >0} x_{(r,i)}^{(t)} \tag{$i\equiv i+C \mod P$}
\\ &= (mcd)^{C} \max_{(r,i)  \in \Sfeeder \cup \Sme } x_{(r,i)}^{(t)}
\end{align*}

Similarly, bounding from below, for all $u\in \Sme$, we have:
\begin{align*}
x_{(q,i)}^{(t+C)} &\ge \frac{1}{mc} \max_{(r,i-1) : M_{(q,i),(r,i-1)} >0} x_{(r,i-1)}^{(t+C-1)}
\\ &\ge \frac{1}{(mc)^2} \max_{(r,i-2): M^2_{(q,i),(r,i-1)} >0} x_{(r,i-2)}^{(t+C-2)}
\\ &\ge\  \dots\ \ge\  \frac{1}{(mc)^{C}} \max_{(r,i) : M^{C}_{(q,i),(r,i)} >0} x_{(r,i)}^{(t)}\tag{$i\equiv i+C \mod P$}
\\ &= \frac{1}{(mc)^{C}} \max_{(r,i)  \in \Sfeeder \cup \Sme } x_{(r,i)}^{(t)}
\end{align*}

Finally we observe that for all $(q,i),(q',i)\in \Sme$ and $t\ge N + C$ we have
\[
\frac{x_{(q,i)}^{(t+C)}}{x_{(q',i)}^{(t+C)}} \le \frac{(mcd)^{C} \displaystyle\max_{(r,i)  \in \Sfeeder \cup \Sme } x_{(r,i)}^{(t)}}{\frac{1}{(mc)^{C}} \displaystyle\max_{(r,i)  \in \Sfeeder \cup \Sme } x_{(r,i)}^{(t)}} = (mc)^{2C}d^{C}.
\]
Hence, by \cref{prop:closeness}, we can select $\eta = \ceil{\log{(mc)^{2C}d^{C}}}+2$
and $N=C$.
\qed
\end{proof}

\section{Proof of Proposition~\ref{prop:ppiscloseness}}

\ppisclose*

\begin{proof}
For $t\ge N$ a pseudo-periodic vector can be expressed as $v_q^{(t)}=  m_q^{(t)} 
10^{\alpha_q^{(t)} + \gamma^{(t)}}$, where $m_q^{(t)}$ is periodic and comes from the finite set of mantissas expressible with $p$ digits, $\alpha^{(t)}_q$ is periodic and $\gamma^{(t)}$ comes from the growth rate and thus does not depend on $q$.

Thus at any given time step we have $\frac{v_q^{(t)}}{v_{q'}^{(t)}}  = \frac{m_q^{(t)}10^{\alpha_q^{(t)} + \gamma^{(t)}}}{m_{q'}^{(t)}10^{\alpha_{q'}^{(t)} + \gamma^{(t)}}} = \frac{m_q^{(t)}10^{\alpha_q^{(t)}}}{m_{q'}^{(t)}10^{\alpha_{q'}^{(t)}}}$. 
Since each of $m_{q},m_{q'},\alpha_{q},\alpha_{q'}$ comes from a finite set of attainable values, the ratio has a maximum $D$ as both $v_q^{(t)}, v_{q'}^{(t)}$ are non-zero. 
Hence 
$v_{q}^{(t)} \approx_\delta v_{q'}^{(t)}$ for $\delta = \lceil\log_{10}(D)\rceil +2$.
\qed
\end{proof}

\section{Proof of Lemma~\ref{lemma:B}}
\label{proofof:lemmaB}
\lemmaB*

\begin{proof}
Note that if $(r,i),(r',i) \in \Sme$ and $(r,i),(r',i) \in\Sfeeder$ then the claim follows from
\cref{lemma:zclosenessofsme} and \cref{prop:ppiscloseness} 
(applied on the sequence of values in the non-zero phase) respectively.

Assume now that $(r,i) \in\Sme$  and $(r',i)\in \Sfeeder$, we show that the claim follows by `transitivity'
of the closeness property (property (3) of \cref{prop:closeness}), due to the closeness within
$\Sfeeder$ and $\Sme$, as well as the closeness implied by the interference.

\noindent More formally,
\begin{itemize}
\item By \cref{prop:ppiscloseness}, as $\Sfeeder$ is pseudo-periodic there exists $\delta$ 
such that for all $t\in \mathbb{N}, (v,i),(v',i) \in \Sfeeder$, 
$x_{(v,i)}^{(t)} \approx_\delta x_{(v',i)}^{(t)}$. 

In particular $x_{(r',i)}^{(t)} \approx_\delta x_{(v,i)}^{(t)}$ for all $(v,i) \in \Sfeeder$.
\item 
By \cref{lemma:zclosenessofsme}, there exists $\eta$ and $K$  
such that for all $t\geq K, (s,i),(s',i) \in \Sme$, 
$x_{(s,i)}^{(t)} \approx_\eta x_{(s',i)}^{(t)}$.

In particular $x_{(r,i)}^{(t)} \approx_\eta x_{(s,i)}^{(t)}$ for all $(s,i) \in \Sme$.
\item Let $t\geq 1$, if $x_{(q,i)}^{(t)}$ is influenced by $(q',i-1)$, then
$\frac 1 m   \leq \frac{x_{(q,i)}^{(t)}}{x_{(q',i-1)}^{(t-1)}} \leq 2 m10^p$.
Moreover, as seen in the proof of \cref{lem:closetop},
$\frac{1}{mc10^{\delta+1}}\leq \frac{x_{(v,i)}^{(t)}}{x_{(q',i-1)}^{t-1}}\leq mcd 10^{\delta +1} $ for some $(v,i)\in \Sfeeder$ such that $M_{(v,i),(u,i-q)} > 0$.
Thus 
$\frac{1}{m^2c10^{\delta+1}} \leq \frac{x_{(q,i)}^{(t)}}{x_{(v,i)}^{(t)}} \leq 2 m^2cd 10^{p+\delta +1}$

Therefore, by setting $\zeta = \ceil{\log_{10}(2 m^2cd)} + p+\delta +3$, 
by property (2) of \cref{prop:closeness} we have 
$x_{(v,i)}^{(t)} \approx_{\zeta} x_{(q,i)}^{(t)}$ for $(q,i)\in \Sme$ and $(v,i)\in \Sfeeder$.
\end{itemize}

\noindent Thus we have  $x_{(r,i)}^{(t)} \approx_{\eta} x_{(q,i)}^{(t)} \approx_\zeta x_{(v,i)}^{(t)}  \approx_\delta x_{(r',i)}^{(t)}$ and by property (3) of \cref{prop:closeness} we have $x_{(r,i)}^{(t)} \approx_{\delta +\eta+ \zeta+8} x_{(r',i)}^{(t)}$.  Thus, the claim holds for $\beta = \delta +\eta+ \zeta + 12$ and
$N = K$.
\qed
\end{proof}

\section{Proof of Lemma~\ref{lemma:finitecase}}
\label{proofof:lemma:finitecase}
\aretheycloseagain*

\begin{proof}

Define a new dynamical system $(M_{\Sme}, y^{})$ such that $y^{} = x^{(t)}_{\Sme}$, with orbit $y^{(t)}$. The vector $y$ evolves without the influence of $\Sfeeder$. Since $y^{(t)}$ consists of a single strongly connected, then it is effectively pseudo-periodic, with starting point $N_y$, period $T_y$ and growth rate $\alpha_y$ for every $q\in \Sme$.

We consider two cases:

First, suppose $y^{(t')} \ne x^{(t+t')}_{\Sme}$ for $t' \le N_y + T\cdot T_y$. Then clearly a value of $x_{\Sfeeder}$ influenced the value of $x_{\Sme}$, and so they must have been close for some $t'$ and we're done.

Secondly, suppose $y^{(t')} = x^{(t+t')}_{\Sme}$ for $t' \le N_y +  T\cdot T_y$ then both $\Sfeeder$ and $\Sme$ completed a synchronised pseudo-period in which they did not interact. We now inspect the increase rate to see if they are converging, so that they will interact in the future, or diverging, so that they will not interact in the future.  This implies that we can detect within $K = N_y + T\cdot T_y$ steps whether $x_{\Sfeeder}^{(t')}$ is close to  $x_{\Sme}^{(t')}$ again. Note that this does entail that $t' \le t + K$ steps as it will take time for the convergence entailed by the growth rates bring them together.

We now analyse the number of steps require until they are close. Consider two states $q\in \Sfeeder$ and $q'\in \Sme$, we have observed that within the first $K$ steps used to determine they will be close again. The two systems can diverge by at most $10^\beta d(mc)^{2K_1}$ in this time (supposing one grows maximally and one reduces maximally at every step). Hence $(x^{(t+K_1)}_{\Sfeeder})_q \approx_\tau (x^{(t+K)}_{\Sme})_{q'}$, where $\tau \le \ceil{\log{10^p 10^\beta d(mc)^{2K}}}$.

Observe we have $\exponent((x^{(t+K)}_{\Sme})_{q'}) > \exponent((x^{(t+K)}_{\Sfeeder})_{q})$, otherwise $\Sfeeder$ would influence $\Sme$. 

By the increase rate of $(x^{(t)}_{\Sme})_{q'}$, every $T_y$ steps, the exponent changes by $\alpha_y$, that is, $\exponent((x^{(t+T_y)}_{\Sme})_{q'}) = \exponent((x^{(t)}_{\Sme})_{q'}) + \alpha_y $.

Similarly by the increase rate of $(x^{(t)}_{\Sfeeder})_{q}$, every $T$ steps, the exponent changes by $\alpha$, that is, $\exponent((x^{(t+T)}_{\Sfeeder})_{q}) = \exponent((x^{(t+T)}_{\Sfeeder})_{q}) + \alpha$.

We observe that exponents become closer at least every $T\cdot T_y$ steps:
\begin{align*}
\exponent((x^{(t+ T\cdot T_y)}_{\Sme})_{q'}) &- \exponent((x^{(t+T\cdot T_y)}_{\Sfeeder})_{q})
\\ &= \exponent((x^{(t)}_{\Sme})_{q'}) +\alpha_y T - \exponent((x^{(t)}_{\Sfeeder})_{q})- \alpha T_y
\\ &= \exponent((x^{(t)}_{\Sme})_{q'}) - \exponent((x^{(t)}_{\Sfeeder})_{q}) +\alpha_y T - \alpha T_y
\\ &< \exponent((x^{(t)}_{\Sme})_{q'}) - \exponent((x^{(t)}_{\Sfeeder})_{q}).
\end{align*}
The final inequality is because $\frac{\alpha_y}{T_y} < \frac{\alpha}{T}$ by the assumption that the exponents are converging. Since the difference reduces by at least one every $T\cdot T_y$ steps, we have $\exponent((x^{(t+K)}_{\Sme})_{q'}) - \exponent((x^{(t+K)}_{\Sme})_{q'})\le p$ within $\tau \cdot T\cdot T_y$ steps. Hence if there exists $t'$ such that $x_{\Sfeeder}^{(t')} \approx_\beta x_{\Sme}^{(t')}$, there exists $t' \le t+ K + \tau\cdot T \cdot T_y$.
\qed \end{proof}

\section{Proof of Lemma~\ref{lemma:eventuallyperiodic}}%

\begin{restatable}{lemma}{eventuallyperiodic}
\label{lemma:eventuallyperiodic}
Let $\target_1,\dots,\target_k$ be sets such that $\mathcal{Z}(\target_i)$ is semi-linear for each $1\le i\le k$. The characteristic word is eventually periodic.
\end{restatable}

\begin{proof}
We show that the characteristic word is eventually periodic. Recall that the alphabet of $w$ is $2^{\{1,\dots,k\}}$.

Let $S\subseteq \{1,\dots, k\}$. By \cref{thm:semialgebraicissemilinear}, since each $\mathcal{Z}(Y_i)$ is semi-linear, observe that the set $C_S = \{i \mid w_i = S\} =\bigcap_{i\in S} \mathcal{Z}(\target_i) \setminus \bigcup_{i\not\in S}\mathcal{Z}(\target_i)$ is semi-linear.
Let $F_S$ be the finite set of $C_S$ and $p_S$ be the common period of it's linear-sets.
Let $F = \max \bigcup_{S\subseteq \{1,\dots, k\}} F_S$ and $p=\operatorname{lcm}_{S\subseteq \{1,\dots, k\}} p_S$.

The word $w$ can thus be represented using an automaton with a finite initial segment of length $F$ and a cycle of length $p$. In the finite initial segment, the $i$-th transition is uniquely labelled by the set S such that $i\in C_S$. The $j$-th character of the cycle is uniquely labelled by the unique S such that $F + j \in C_S$.
\qed \end{proof}

\end{document}